\documentclass[journal, pdftex]{IEEEtran}

\usepackage{amsmath} 
\usepackage{amssymb} 
\usepackage{graphicx} 
\usepackage{array} 
\usepackage{booktabs} 
\usepackage{pgfplotstable} 
\usepackage{cite} 

\usepackage{pgfplots}
\pgfplotsset{compat=1.18}

\usepackage{amsthm}

\theoremstyle{plain}
\newtheorem{theorem}{Theorem}
\newtheorem{lemma}[theorem]{Lemma}
\newtheorem{assumption}{Assumption}
\newtheorem{remark}{Remark}

\usepackage[utf8]{inputenc}
\usepackage[T1]{fontenc}
\usepackage[margin=1in]{geometry}
\usepackage{bm}

\usepackage{algorithm}
\usepackage{algorithmicx}
\usepackage{algpseudocode}

\ifCLASSINFOpdf
  \graphicspath{{./}} 
  \DeclareGraphicsExtensions{.pdf,.jpeg,.png}
\else
\fi

\begin{document}
%
\title{SDC-Based Model Predictive Control: Enhancing Computational Feasibility for Safety-Critical Quadrotor Control}
%
%
\author{Saber~Omidi,~\IEEEmembership{Student Member,~IEEE}
\thanks{S. Omidi is a Ph.D. Candidate with the Department of Mechanical Engineering, University of New Hampshire, Durham, NH 03824 USA (e-mail: saber.omidi@unh.edu).}
\thanks{Manuscript received [Date]; revised [Date].}}

\markboth{IEEE Transactions on Automatic Control,~Vol.~, No.~, Month~Year}%
{Omidi: SDC-Based MPC for Safety-Critical Quadrotor Control}

\maketitle

\begin{abstract}
Nonlinear Model Predictive Control (NMPC) is widely used for controlling high-speed robotic systems such as quadrotors. However, its significant computational demands often hinder real-time feasibility and reliability, particularly in environments requiring robust obstacle avoidance. This paper proposes a novel SDC-Based Model Predictive Control (MPC) framework, which preserves the high-precision performance of NMPC while substantially reducing computational complexity by over 30\%. By reformulating the nonlinear quadrotor dynamics through the State-Dependent Coefficient (SDC) method, the original nonlinear program problem is transformed into a sequential quadratic optimization problem. The controller integrates an integral action to eliminate steady-state tracking errors and imposes constraints for safety-critical obstacle avoidance. Additionally, a disturbance estimator is incorporated to enhance robustness against external perturbations. Simulation results demonstrate that the SDC-Based MPC achieves comparable tracking accuracy to NMPC, with greater efficiency in terms of computation times, thereby improving its suitability for real-time applications. Theoretical analysis further establishes the stability and recursive feasibility of the proposed approach.
\end{abstract}


\IEEEpeerreviewmaketitle


\section{Introduction}\label{sec:introduction}
\IEEEPARstart{U}{nmanned} aerial vehicles (UAVs), particularly quadrotors, have become increasingly important in both academic research and industry \cite{peksa2024review}. This growth is mainly due to their low cost, simple design, and ability to move with agility \cite{allan2025low}. Quadrotors can take off and land vertically and hover in place, which makes them well-suited for tasks that traditional aircraft cannot handle \cite{wang2025autonomous}. They are now used in areas like photography, surveillance, search and rescue, precision agriculture, and package delivery \cite{unde2025expanding}. As a result, improving flight control systems is essential to make the most of these vehicles in real-world situations \cite{rinaldi2023comparative}.
Effectively controlling quadrotors presents a unique challenge due to their complex, nonlinear, and underactuated dynamics. Over the years, researchers and engineers have explored a wide range of flight control strategies to address these difficulties. Techniques, such as the Proportional-Integral-Derivative (PID) controller, remain popular for basic flight stabilization because they are straightforward to implement and tune~\cite{masse2018modeling}. Nevertheless, PID controllers depend on simplified, linearized models and can struggle to maintain performance during aggressive maneuvers or when faced with unexpected disturbances. To achieve better control in such situations, optimal controllers like the Linear Quadratic Regulator (LQR) are used to systematically compute control gains that minimize a predefined cost function. However, LQR and similar methods are typically effective only under nominal, linearized conditions~\cite{chen2021feedback}.
Recognizing these shortcomings, nonlinear control approaches have gained traction for their ability to deal with the system's true dynamics. For example, backstepping uses a stepwise design to guarantee stability, while sliding mode control offers strong resilience to model uncertainties and external disturbances by guiding the system along a specific trajectory~\cite{wang2023adaptive, mu2017integral}. The aforementioned control methods are usually organized in a cascade control structure. This approach simplifies the control problem by separating it into a rapid inner loop for attitude stabilization and a slower outer loop for position tracking~\cite{bouabdallah2004pid}.
In recent years, learning-based controllers—particularly those leveraging reinforcement learning and neural networks—have shown promise by learning optimal control policies directly from real-world experience, eliminating the need for precise mathematical models~\cite{hwangbo2017control}. Despite their potential, these methods often require large datasets for training and pose challenges in ensuring safe and reliable operation.
Ensuring the safety of autonomous systems is a critical consideration in their synthesis and implementation, which has motivated the development of formally verifiable control methodologies. A prominent approach involves the use of Control Barrier Functions (CBFs), which are employed to render a defined safe set forward invariant. A CBF-based controller minimally modifies a nominal performance-oriented law; this intervention occurs only as the system state approaches the boundary of the safe set, thereby guaranteeing the system does not enter the unsafe region~\cite{ames2017control}. For stronger safety guarantees, albeit often at a higher computational cost, Reachability Analysis is utilized. This methodology entails the computation of the reachable set—the set of all states attainable by the system—and subsequently verifying that this set does not intersect with predefined unsafe regions~\cite{mitchell2005toolbox}. Another influential paradigm leverages Lyapunov-based techniques, which are conventionally applied to stability analysis. In this context, a Lyapunov-like function, referred to as a barrier certificate, is constructed such that its value defines a safe sublevel set of the state space. The control law is then synthesized to ensure the time derivative of the barrier certificate is non-positive, thereby formally confining the system's evolution to this region~\cite{prajna2007framework}.
Model Predictive Control (MPC) is an advanced control framework that leverages an explicit system model to predict future behavior over a finite time horizon. At each sampling instant, MPC solves an online optimization problem to find an optimal sequence of control inputs, after which only the first input is applied in a strategy known as a receding horizon. Its primary strength lies in its ability to explicitly incorporate system constraints, making it highly suitable for safety-critical applications. The nature of the underlying optimization is dictated by the system model and constraints; linear systems with quadratic costs result in a solvable Quadratic Program (QP), while fully nonlinear systems require solving a computationally intensive Nonlinear Program (NLP).
Building upon this core framework, several powerful variants have been developed. Nonlinear MPC (NMPC) directly uses a nonlinear model for high-fidelity prediction at the cost of significant online computation. To address this computational burden, Explicit MPC (EMPC) pre-computes the optimal control law offline as a piecewise affine function of the state, reducing the online effort to a simple function evaluation but is limited to lower-dimensional systems~\cite{bemporad2002explicit}. To handle real-world uncertainties, Robust MPC provides worst-case guarantees. A prominent implementation is Tube MPC, where a nominal trajectory is optimized and a separate feedback controller ensures the actual system state remains within a robust "tube" around this path, containing all possible deviations due to uncertainty. For probabilistic disturbances, Stochastic MPC (SMPC) optimizes expected performance while satisfying chance constraints~\cite{farina2016stochastic}. Furthermore, modern Learning-based MPC integrates machine learning to learn system dynamics from data, adapting its predictive model over time~\cite{rosolia2017learning}.
Despite its versatility, MPC's performance is fundamentally reliant on the accuracy of the predictive model, and its significant online computational demand can render it unsuitable for systems with very fast dynamics or limited onboard processing power.
This paper addresses the critical challenge of developing safe, robust, and computationally tractable control systems for quadrotor unmanned aerial vehicles (UAVs). A comprehensive simulation framework was established in MATLAB, utilizing CasADi for numerical optimization, to systematically evaluate three distinct Model Predictive Control (MPC) architectures. The controllers were benchmarked against a challenging trajectory tracking mission requiring static obstacle avoidance amidst stochastic external disturbances. The evaluated strategies include a high-fidelity Nonlinear MPC (NMPC), which solves the full nonlinear optimal control problem and serves as a performance benchmark, alongside two novel quasi-linear formulations. The primary contribution of this research is the development of a robust, computationally efficient MPC that leverages a State-Dependent Coefficient (SDC) representation of the quadrotor's nonlinear dynamics. This technique reformulates the system equations into a pseudo-linear form, which enables the nonlinear optimal control problem to be effectively approximated as a Quadratic Program (QP) at each sampling instant, significantly reducing computational complexity compared to the full NLP. Building upon this efficient foundation, we introduce a Robust SDC-based MPC (R-SDC-MPC). This advanced controller integrates a real-time disturbance observer that estimates lumped model uncertainties and external disturbances from the one-step state prediction error. This disturbance estimate is then used to actively compensate the predictive model, enhancing the controller's resilience. Through quantitative analysis of tracking accuracy (MSE), computational load (CPU time), and overall cost, our work demonstrates that the proposed R-SDC-MPC offers a compelling trade-off, achieving robustness and performance comparable to NMPC while maintaining the computational tractability of a QP-based approach.

The remainder of this paper is organized as follows. Section~\ref{sec:dynamics} presents the nonlinear dynamical model of the quadrotor and details its reformulation into the State-Dependent Coefficient (SDC) representation. Section~\ref{sec:design} develops the proposed control framework, starting with the NMPC benchmark, followed by the nominal SDC-MPC design, and culminating in the robust controller with disturbance estimation. In Section~\ref{sec:results}, a comprehensive simulation study is presented to validate the performance of the proposed methods, with a focus on tracking accuracy, computational efficiency, and safety. Finally, Section~\ref{sec:conclusion} concludes the paper with a summary of the key findings and outlines directions for future research.

\section{Quadrotor Dynamics and Problem Formulation}\label{sec:dynamics}

\subsection{Nonlinear Quadrotor Model}\label{ssec:model}
In this work, we adopt the six-degree-of-freedom (6-DoF) rigid-body model for the quadrotor, as formulated in~\cite{masse2018modeling,xu2023multi}. Figure \ref{fig:quad_model} shows the 6-DoF quadrotor model, illustrating the inertial and body-fixed frames, the primary forces of thrust and gravity, and the control torques for roll, pitch, and yaw.
\begin{figure}[htbp]
    \centering
    \includegraphics[width=\linewidth]{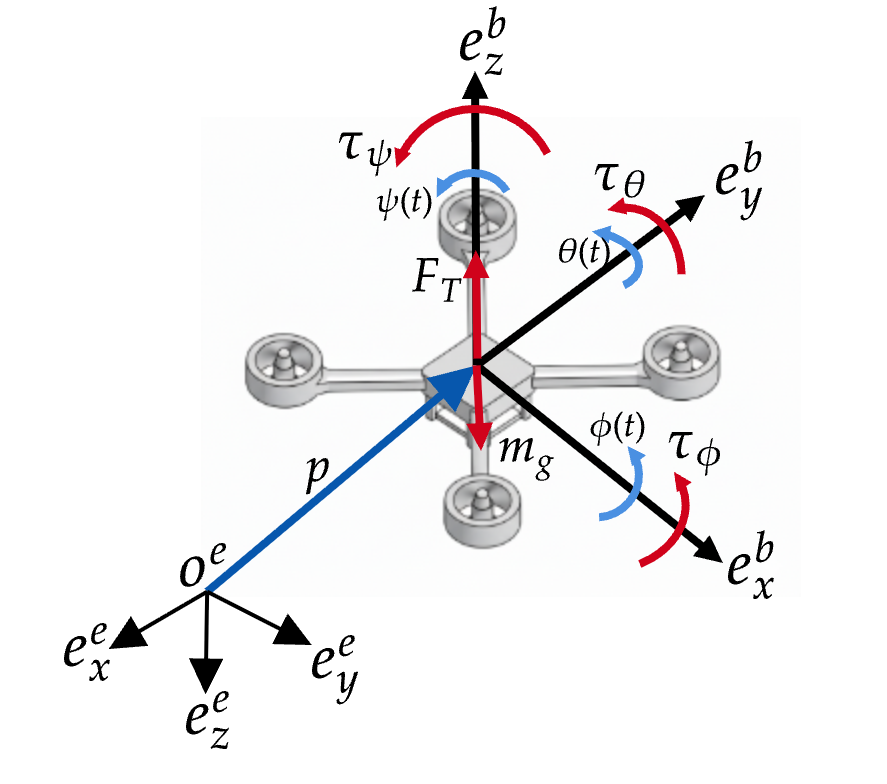}
    \caption{Schematic of the 6-DoF quadrotor model, illustrating the inertial frame $\{e\}$ and the body-fixed frame $\{b\}$. The quadrotor's position is described by the vector $p$. Control inputs consist of the total thrust force $F_T$ along the $e_z^b$ axis and the body torques $\tau_\phi$ (roll), $\tau_\theta$ (pitch), and $\tau_\psi$ (yaw), which induce the positive rotations $\phi(t)$, $\theta(t)$, and $\psi(t)$, respectively. The gravitational force $m_g$ acts in the negative $e_z^e$ direction.}
    \label{fig:quad_model}
\end{figure}
The 12-dimensional state vector $\mathbf{x} \in \mathbb{R}^{12}$ is defined as:
\begin{equation}
    \mathbf{x} = [\mathbf{p}^T, \mathbf{v}^T, \bm{\eta}^T, \bm{\omega}^T]^T
\end{equation}
where $\mathbf{p} = [x, y, z]^T$ is the position in the inertial frame, $\mathbf{v} = [v_x, v_y, v_z]^T$ is the velocity in the inertial frame, $\bm{\eta} = [\phi, \theta, \psi]^T$ represents the Euler angles (roll, pitch, yaw), and $\bm{\omega} = [p, q, r]^T$ is the angular velocity vector in the body frame. The control input vector $\mathbf{u} \in \mathbb{R}^{4}$ consists of the total thrust $T$ and the three body-axis torques $\bm{\tau} = [\tau_\phi, \tau_\theta, \tau_\psi]^T$:
\begin{equation}
    \mathbf{u} = [T, \bm{\tau}^T]^T
\end{equation}

The translational and rotational dynamics are governed by the Newton-Euler equations:
\begin{align}
    \dot{\mathbf{p}} &= \mathbf{v} \\
    \dot{\mathbf{v}} &= \mathbf{g} + \frac{1}{m} \mathbf{R}(\bm{\eta}) \mathbf{F}_T \\
    \dot{\bm{\eta}} &= \mathbf{W}(\bm{\eta}) \bm{\omega} \\
    \dot{\bm{\omega}} &= \mathbf{J}^{-1} (\bm{\tau} - \bm{\omega} \times (\mathbf{J}\bm{\omega}))
\end{align}
where $m$ is the total mass, $\mathbf{J} \in \mathbb{R}^{3 \times 3}$ is the inertia matrix, and $\mathbf{g} = [0, 0, -g]^T$ is the acceleration due to gravity. The thrust vector in the body frame is $\mathbf{F}_T = [0, 0, T]^T$. The matrix $\mathbf{R}(\bm{\eta}) \in SO(3)$ is the rotation matrix mapping vectors from the body frame to the inertial frame, and $\mathbf{W}(\bm{\eta})$ is the kinematic transformation matrix from body rates to Euler angle rates.

\subsection{SDC Representation}\label{ssec:sdc_rep}

A general continuous-time nonlinear dynamical system can be expressed as:
\begin{equation}
    \dot{\mathbf{x}}(t) = \mathbf{f}(\mathbf{x}(t), \mathbf{u}(t)), \qquad \mathbf{x}(0) = \mathbf{x}_0,
    \tag{5}
\end{equation}
where $\mathbf{x}(t) \in \mathbb{R}^n$ and $\mathbf{u}(t) \in \mathbb{R}^m$ are the system state vector and the control input vector, respectively; $\mathbf{x}_0 \in \mathbb{R}^n$ is the initial-valued vector of the states. For the purpose of control design, a possible way to express the dynamical system in (5) is as:
\begin{align}
    \dot{\mathbf{x}}(t) &= \mathbf{a}(\mathbf{x}(t)) + \mathbf{B}(\mathbf{x}(t))\mathbf{u}(t), \tag{6a} \\
    \mathbf{y}(t) &= \mathbf{c}(\mathbf{x}(t)) \tag{6b}
\end{align}
where $\mathbf{y}(t) \in \mathbb{R}^q$ is the system output vector; $\mathbf{a}(\mathbf{x}(t)) : \mathbb{R}^n \to \mathbb{R}^n$, $\mathbf{B}(\mathbf{x}(t)) : \mathbb{R}^n \to \mathbb{R}^{n \times m}$, and $\mathbf{c}(\mathbf{x}(t)) : \mathbb{R}^n \to \mathbb{R}^q$ are assumed to be smooth. The matrix-valued functions in (6) are assumed to be smooth and satisfy $\mathbf{a}(0) = \mathbf{c}(0) = 0$. Therefore, these functions can be written in their pseudo-linearized or SDC forms:
\begin{align}
    \mathbf{a}(\mathbf{x}(t)) &= \mathbf{A}(\mathbf{x}(t))\mathbf{x}(t), \tag{8a} \\
    \mathbf{c}(\mathbf{x}(t)) &= \mathbf{C}(\mathbf{x}(t))\mathbf{x}(t) \tag{8b}
\end{align}
where $\mathbf{A} : \mathbb{R}^n \to \mathbb{R}^{n \times n}$, and $\mathbf{C} : \mathbb{R}^n \to \mathbb{R}^{p \times n}$ are two matrix-valued functions. It is important to note that there are infinite ways to write pseudo-linearized non-scalar systems which allow an additional level of flexibility to select [25]. The specific factorization used in this work is chosen for its straightforward derivation from the Newton-Euler equations and its ability to isolate the control inputs in the $\mathbf{B}(\mathbf{x})$ matrix, which is advantageous for control design. Furthermore, this structure has been shown to maintain the necessary stabilizability and detectability properties for a wide range of operating conditions.
Now that the matrices are expressed in SDC form, the system can be written as follows:
\begin{align}
    \dot{\mathbf{x}}(t) &= \mathbf{A}(\mathbf{x}(t))\mathbf{x}(t) + \mathbf{B}(\mathbf{x}(t))\mathbf{u}(t), \tag{9a} \\
    \mathbf{y}(t) &= \mathbf{C}(\mathbf{x}(t))\mathbf{x}(t). \tag{9b}
\end{align}
we consider a full state-feedback control problem, where the entire state vector $\mathbf{x}(t)$ is assumed to be available to the controller at each time instant. Consequently, the system output is the state itself, rendering a separate output equation unnecessary for the control formulation. The dynamics of the quadrotor are therefore represented by the following SDC-based pseudo-linear model:
\begin{equation}
    \mathbf{B}(\mathbf{x})  = 
    \begin{bmatrix}
        \mathbf{0}_{3 \times 1} & \mathbf{0}_{3 \times 3} \\
        \frac{1}{m}\mathbf{R}(\bm{\eta})\mathbf{e}_3 & \mathbf{0}_{3 \times 3} \\
        \mathbf{0}_{3 \times 1} & \mathbf{0}_{3 \times 3} \\
        \mathbf{0}_{3 \times 1} & \mathbf{J}^{-1}
    \end{bmatrix}
\end{equation}
where $\mathbf{e}_3 = [0, 0, 1]^T$. The system matrix $\mathbf{A}(\mathbf{x}) \in \mathbb{R}^{12 \times 12}$ is:
\begingroup 
\setlength{\arraycolsep}{3pt} 
\begin{equation}
    \mathbf{A}(\mathbf{x}) = 
    \begin{bmatrix}
        \mathbf{0}_{3 \times 3} & \mathbf{I}_{3 \times 3} & \mathbf{0}_{3 \times 3} & \mathbf{0}_{3 \times 3} \\
        \mathbf{0}_{3 \times 3} & \mathbf{0}_{3 \times 3} & \mathbf{A}_{23}(\bm{\eta}, T) & \mathbf{0}_{3 \times 3} \\
        \mathbf{0}_{3 \times 3} & \mathbf{0}_{3 \times 3} & \mathbf{A}_{33}(\bm{\eta}, \bm{\omega}) & \mathbf{W}(\bm{\eta}) \\
        \mathbf{0}_{3 \times 3} & \mathbf{0}_{3 \times 3} & \mathbf{0}_{3 \times 3} & \mathbf{A}_{44}(\bm{\omega})
    \end{bmatrix}
\end{equation}
\endgroup
where the non-zero, state-dependent blocks are defined as:
\begin{align}
    \mathbf{A}_{23}(\bm{\eta}, T) &= \frac{T}{m}\frac{\partial \mathbf{R}(\bm{\eta})}{\partial \bm{\eta}}, \\
    \mathbf{A}_{33}(\bm{\eta}, \bm{\omega}) &= \frac{\partial (\mathbf{W}(\bm{\eta})\bm{\omega})}{\partial \bm{\eta}}, \\
    \mathbf{A}_{44}(\bm{\omega}) &= \mathbf{J}^{-1}(\mathbf{S}(\mathbf{J}\bm{\omega}) - \mathbf{S}(\bm{\omega})\mathbf{J}).
\end{align}
Here $\mathbf{I}_{3 \times 3}$ is the identity matrix and $\mathbf{0}_{3 \times 3}$ is zero matrix. $\frac{\partial \mathbf{R}(\bm{\eta})}{\partial \bm{\eta}}$ is the partial derivative of the rotation matrix with respect to the Euler angles. $\frac{\partial (\mathbf{W}(\bm{\eta})\bm{\omega})}{\partial \bm{\eta}}$ is the partial derivative of the kinematic transformation. $\mathbf{S}(\cdot)$ is the skew-symmetric matrix operator. For a comprehensive review of related quasi-linear control methodologies, including the closely allied State-Dependent Riccati Equation (SDRE) framework, the reader is referred to the works in~\cite{cimen2008sdre,cloutier1997sdre,mracek1998control,hammett2016sdc}.

\subsection{SDC-Based Optimal Control Problem Formulation}\label{ssec:sdc_ocp}

The main idea of the proposed SDC-based MPC framework is an online optimization problem solved at each sampling instant \(t\). This problem aims to compute an optimal control input sequence over a finite prediction horizon, \(T_p\), that steers the quadrotor along a time-varying reference trajectory, \(\mathbf{x}_{\text{ref}}(\tau)\). Given the current state \(\mathbf{x}(t)\), the finite-horizon optimal control problem is formulated as follows:
\begin{subequations} \label{eq:trajectory_optimization}
\begin{align}
    \min_{\mathbf{u}(\cdot)} J(.) &\!=\! \int_{t}^{t+T_p} \!\left(\! \|\mathbf{x}(\tau) \!-\! \mathbf{x}_{\text{ref}}(\tau)\|^2_{\mathbf{Q}} \!+\! \|\mathbf{u}(\tau)\|^2_{\mathbf{R}}\right)\! d\tau \nonumber \\
    &+ \|\mathbf{x}(t+T_p) - \mathbf{x}_{\text{ref}}(t+T_p)\|^2_{\mathbf{P}(\mathbf{x})} \label{eq:cost_function} \\
    \text{s.t.}\quad \dot{\mathbf{x}}(\tau) &=\mathbf{A}(\mathbf{x}(\tau))\mathbf{x}(\tau) + \mathbf{B}(\mathbf{x}(\tau))\mathbf{u}(\tau), \label{eq:sdc_dynamics} \\
    \mathbf{x}(\tau) &\in \mathbb{X}, \label{eq:state_constraint} \\
    \mathbf{u}(\tau) &\in \mathbb{U}, \label{eq:input_constraint} \\
    \mathbf{x}(t+T_p) &\in \mathbb{X}_T(t), \label{eq:terminal_constraint}
    \quad \forall \tau \in [t, t+T_p]
\end{align}
\end{subequations}
where the optimization is performed over the control input sequence \(\mathbf{u}(\cdot)\) for all \(\tau \in [t, t+T_p]\). The objective function in \eqref{eq:cost_function} is composed of a running cost and a terminal cost. The integral term represents the running cost, which penalizes the deviation of the predicted state trajectory \(\mathbf{x}(\tau)\) from the reference \(\mathbf{x}_{\text{ref}}(\tau)\) as well as the magnitude of the control input \(\mathbf{u}(\tau)\). The semi-positive-definite and positive-definite weighting matrices \(\mathbf{Q}\in \mathbb{R}^{n \times n}\) and \(\mathbf{R} \in \mathbb{R}^{m \times m}\), respecetivly, allow for tuning the trade-off between tracking accuracy and control effort. The terminal cost, weighted by the matrix \(\mathbf{P}(\mathbf{x})\), penalizes the final state error and is crucial for proving closed-loop stability.
The optimization is subject to several constraints that define the system's behavior and limitations.
\begin{itemize}
    \item The differential equation \eqref{eq:sdc_dynamics} enforces the system's SDC dynamics throughout the prediction horizon.
    \item Constraints \eqref{eq:state_constraint} and \eqref{eq:input_constraint} confine the state and control inputs to their respective admissible sets. The input constraint set, \(\mathbb{U}\), directly represents the physical limitations of the actuators (e.g., maximum thrust and torques). The state constraint set, \(\mathbb{X}\), is paramount for defining the safe operating envelope. For safety-critical tasks such as obstacle avoidance, this set is explicitly structured to exclude unsafe regions of the state-space. By imposing such constraints, the MPC controller is forced to find a control sequence that guarantees the predicted trajectory remains entirely within the safe region.
    \item Finally, the terminal inequality constraint\eqref{eq:terminal_constraint} compels the predicted state to lie within a terminal set \(\mathbb{X}_T(t)\) at the end of the horizon. This set is typically defined relative to the reference state \(\mathbf{x}_{\text{ref}}(t+T_p)\) and is essential for guaranteeing the recursive feasibility and stability of the closed-loop system.
\end{itemize}
Finally, a primary objective of this work is to improve the computational efficiency of the online optimization, ensuring the controller is feasible for real-time implementation without sacrificing performance or violating system constraints.

\section{Design of the SDC-Based MPC Framework}\label{sec:design}

\subsection{NMPC Benchmark}\label{ssec:nmpc}
The NMPC benchmark is based on the quasi-infinite horizon scheme proposed in \cite{chen1998quasi}. At each sampling instant \(t\), the controller solves the following finite-horizon optimal control problem for the current state \(x(t)\):
\begin{subequations} \label{eq:nmpc_benchmark}
\begin{align}
    \min_{\mathbf{u}(\cdot)} J(.) &\!=\! \int_{t}^{t+T_p} \!\left(\! \|\mathbf{x}(\tau) \!-\! \mathbf{x}_{\text{ref}}(\tau)\|^2_{\mathbf{Q}} \!+\! \|\mathbf{u}(\tau)\|^2_{\mathbf{R}}\right)\! d\tau \nonumber \\
    &+ \|\mathbf{x}(t+T_p) - \mathbf{x}_{\text{ref}}(t+T_p)\|^2_{\mathbf{Q}_t} \label{eq:cost_function_n} \\
    \text{s.t.}\quad \dot{\mathbf{x}}(\tau) &= \mathbf{f}(\mathbf{x}(\tau), \mathbf{u}(\tau)),
    \label{eq:n_dynamics} \\
    \mathbf{x}(\tau) &\in \mathbb{X}, \label{eq:state_constraint_n} \\
    \mathbf{u}(\tau) &\in \mathbb{U}, \label{eq:input_constraint_n} \\
    \mathbf{x}(t+T_p) &\in \Omega, \label{eq:terminal_constraint_n}
    \quad \forall \tau \in [t, t+T_p]
\end{align}
\end{subequations}
where the terminal penalty matrix \(\mathbf{Q}_t \in \mathbb{R}^{n \times n}\) is a positive-definite, symmetric matrix chosen to guarantee stability \cite{chen1998quasi}. It is computed off-line as the unique solution to the Lyapunov equation:
\begin{align}
        (\mathbf{A}_K + \kappa \mathbf{I})^T \mathbf{S} \!+\! \mathbf{S} (\mathbf{A}_K + \kappa \mathbf{I})\!=\!-(\mathbf{Q} + \mathbf{K}^T \mathbf{R} \mathbf{K})    
\end{align}
where \(\mathbf{A}_K = \mathbf{A}+\mathbf{BK}\) is the closed-loop matrix of the system's Jacobian linearization, and \(\kappa \geq 0\) is a tuning parameter \cite{chen1998quasi}. \(\mathbf{K}\) is a locally stabilizing state feedback gain determined linearized system, typically by solving the LQR problem \cite{mare2007solution}. The terminal region \(\Omega\) is an invariant set for the nonlinear system under a local linear state feedback law \(\mathbf{u}=\mathbf{K}\mathbf{x}(t)\) \cite{chen1998quasi}. It is defined as a level set of the Lyapunov function, \(\Omega := \{ \mathbf{x} \in \mathbb{R}^n \mid \mathbf{x}^T \mathbf{S} \mathbf{x} \leq \alpha \}\), where \(\alpha > 0\) is chosen to ensure that for any state within \(\Omega\), the corresponding linear control input satisfies the constraint \(\mathbf{K}\mathbf{x}(t) \in \mathbb{U}\) \cite{chen1998quasi}.
The explicit use of a general nonlinear dynamics model, fundamentally casts the optimal control problem as a non-convex NLP. Solving such an NLP is computationally demanding and challenging when the predicted trajectory interacts with the boundaries of the constraint sets. This is particularly for safety-critical applications where complex state constraints, \(\mathbf{x}(\tau) \in \Omega\), are activated to enforce collision avoidance by proscribing certain regions of the state-space. The computational burden can render real-time implementation infeasible. 
Moreover, the reliance on a local linearization to prove stability introduces significant conservatism for systems with prominent nonlinearities. The terminal region where the linear approximation is valid can be impractically small \cite{chen1998quasi}. This necessitates a longer prediction horizon, \(T_p\), to ensure the optimization problem remains feasible, which in turn increases the computational burden and can severely restrict the controller's effective region of attraction.

\subsection{Nominal SDC-Based MPC Formulation}\label{ssec:nominal_sdc}

To address the computational challenges inherent to the NMPC benchmark, the nominal SDC-MPC scheme is proposed. The primary advantage of this approach is its ability to reformulate the computationally intensive NLP into a sequence of QP with utilization of the SDC representation into the dynamic constraint. By parameterizing the system matrices with the most recent state measurement, the dynamic constraint for the optimal control problem is rendered as a affine system. This model is unique to each time step, as the matrices \(\mathbf{A}(\mathbf{x})\) and \(\mathbf{B}(\mathbf{x})\) are re-calculated based on the new state measurement. This transformation of the dynamics into a set of linear equality constraints is critical. When combined with a quadratic objective function and convex state and input constraints, it ensures that the overall optimal control problem becomes a convex QP. This QP can be solved efficiently and reliably, avoiding the complexities and computational burden of the original NLP.
Another key consequence of SDC formulation is that the terminal penalty matrix, \(\mathbf{P}(\mathbf{x})\), also becomes state-dependent. Consequently, a unique terminal penalty matrix, \(\mathbf{P}(\mathbf{x})\), must be computed at each time step based on the current state measurement to ensure stability of the system. With considering this, at each state \(\mathbf{x}(t)\), the State-Dependent Lyapunov Equation (SDLE) is formulated to find the positive-definite matrix \(\mathbf{P}(\mathbf{x})\) that satisfies:
\begin{align}
    \big(\mathbf{A}&(\mathbf{x}) - \mathbf{B}(\mathbf{x})\mathbf{K}(\mathbf{x})\big)^T \mathbf{P}(\mathbf{x}) \nonumber \\
     + \mathbf{P}(\mathbf{x}) \big(\mathbf{A}&(\mathbf{x}) - \mathbf{B}(\mathbf{x})\mathbf{K}(\mathbf{x})\big) \nonumber \\ +  \big(\mathbf{Q}&(\mathbf{x}) +  \mathbf{K}(\mathbf{x})^T \mathbf{R} \mathbf{K}(\mathbf{x})\big) = \mathbf{0}
\label{eq:state_dependent_lyapunov}
\end{align}
where the existence of a positive-definite solution \(\mathbf{P}(\mathbf{x})\) for all \(\mathbf{x}\) in a given domain is a sufficient condition for local stability. The existence of a unique, positive-definite solution \(\mathbf{P}(\mathbf{x})\) to SDLE is contingent upon the stability of the closed-loop system, which is directly determined by the feedback gain matrix \(\mathbf{K}(\mathbf{x})\). This gain is the sub-optimal state feedback controller, computed as 
\begin{align}
\mathbf{K}(\mathbf{x}) = \mathbf{R}^{-1} \mathbf{B}(\mathbf{x})^T \mathbf{S}(\mathbf{x}),    
\end{align}
 and is designed to stabilize the system at the state \(\mathbf{x}\). The gain itself is derived from the solution \(\mathbf{S}(\mathbf{x})\) of the State-Dependent Riccati Equation (SDRE), given by:
\begin{align}
    &\mathbf{A}(\mathbf{x})^T \mathbf{S}(\mathbf{x}) + \mathbf{S}(\mathbf{x}) \mathbf{A}(\mathbf{x})\nonumber \\ - &\mathbf{S}(\mathbf{x}) \mathbf{B}(\mathbf{x}) \mathbf{R}^{-1} \mathbf{B}(\mathbf{x})^T \mathbf{S}(\mathbf{x}) + \mathbf{Q}(\mathbf{x}) = \mathbf{0}
    \label{eq:continuous_sdre_full}
\end{align}
Therefore, the existence of a positive-definite solution to the Lyapunov equation depends on finding a stabilizing solution to the Riccati equation. While solving the SDRE and SDLE at each time step introduces a computational overhead, this process is typically very fast for systems of this dimension and its cost is included in the total per-iteration CPU times reported in Section~\ref{sec:results}. For the quadrotor model, this overhead was found to be a small fraction of the total optimization time.
Under the following theorem, we can guarantee that the SDRE solution is feasible \cite{cimen2008sdre}.

\begin{theorem}[Existence of a Stabilizing SDRE Solution]
For a nonlinear system in SDC form, a unique, positive-definite stabilizing solution \(\mathbf{S}(\mathbf{x})\) to the State-Dependent Riccati Equation \eqref{eq:continuous_sdre_full} exists at a state \(\mathbf{x}\) if the following state-dependent pairs are pointwise stabilizable and detectable, respectively \cite{cimen2008sdre,batmani2016nonlinear}:
\begin{enumerate}
    \item The pair \((\mathbf{A}(\mathbf{x}), \mathbf{B}(\mathbf{x}))\) is \textbf{stabilizable}.
    \item The pair \((\mathbf{A}(\mathbf{x}), \mathbf{Q}(\mathbf{x})^{1/2})\) is \textbf{detectable}.
\end{enumerate}
\end{theorem}

The fundamental argument for establishing recursive feasibility is conceptually consistent for SDC based MPC and linear MPC schemes. This argument relies on constructing a feasible candidate solution for the current time step using the optimal solution from the preceding step. Therefore, the proof of recursive feasibility is contingent upon a key assumption regarding the system's local dynamics and a constructive lemma that guarantees the existence of the necessary terminal ingredients. Nevertheless, the theoretical underpinnings of the proof and its practical implications diverge significantly due to the state-dependent, nonlinear nature of the SDC framework.

\begin{assumption}[Stabilizability of the Linearization] \label{as:stabilizability}
The Jacobian linearization of the system at the origin, represented by the pair of matrices \((\mathbf{A}, \mathbf{B})\), is stabilizable \cite{chen1998quasi}.
\end{assumption}

This assumption is fundamental, as it ensures that the system can be locally stabilized, which is a prerequisite for constructing a valid terminal set for the MPC. This leads to the following critical lemma.

\begin{lemma}[Existence of Terminal Ingredients] \label{lem:terminal}
If Assumption \ref{as:stabilizability} holds, then a stabilizing linear feedback gain \(\mathbf{K}\), a positive-definite terminal penalty matrix \(\mathbf{P}\), and a terminal region \(\mathbb{X}_T\) can be constructed. This terminal region \(\mathbb{X}_T\) is a \textbf{positively invariant set} for the full nonlinear system under the local linear feedback law \(\mathbf{u} = \mathbf{Kx}\) \cite{chen1998quasi, rawlings2017model}.
\end{lemma}
With this foundation, we can formally state and prove the recursive feasibility of the controller.
\begin{theorem}[Recursive Feasibility]
If Assumption \ref{as:stabilizability} holds and the optimization problem is feasible at the initial time \(t=0\), then a feasible solution is guaranteed to exist for all subsequent time steps \(t > 0\).
\end{theorem}
\begin{proof}
The proof is by induction. Assume a feasible optimal solution exists at time \(t\), yielding the control sequence \(\mathbf{u}^*(\cdot; t)\) and state trajectory \(\mathbf{x}^*(\cdot; t)\), which satisfies the terminal constraint \(\mathbf{x}^*(t+T_p; t) \in \mathbb{X}_T\). At the next sampling instant, \(t+\delta\), we construct a candidate solution \(\hat{\mathbf{u}}(\cdot)\) for the new horizon by using the tail of the previous solution and appending the terminal control law:
\begin{equation}
\hat{\mathbf{u}}(\tau) =
\begin{cases}
    \mathbf{u}^*(\tau; t) & \text{for } \tau \in [t+\delta, t+T_p] \\
    \mathbf{K} \hat{\mathbf{x}}(\tau) & \text{for } \tau \in (t+T_p, t+\delta+T_p]
\end{cases}
\end{equation}
where the existence of the stabilizing gain \(\mathbf{K}\) and the terminal set \(\mathbb{X}_T\) is guaranteed by Lemma \ref{lem:terminal}. This candidate is feasible because the terminal state \(\mathbf{x}^*(t+T_p; t)\) lies in \(\mathbb{X}_T\), which is a positively invariant set under the control law \(\mathbf{u}=\mathbf{Kx}\). Therefore, the new terminal state at \(t+\delta+T_p\) is also guaranteed to be in \(\mathbb{X}_T\), satisfying the terminal constraint. Since a feasible solution can be constructed, the optimization problem remains feasible for all subsequent steps \cite{chen1998quasi, rawlings1993stability}.
\end{proof}
The determination of a locally stabilizing feedback gain, derived from the Jacobian linearization of the system at its equilibrium, constitutes the foundational prerequisite for the recursive feasibility proof. This concept is illustrated conceptually in Figure~\ref{fig:recursive_feasibility}. This off-line synthesis of a stabilizing terminal controller and an associated invariant terminal set enables the construction of a feasible candidate solution at each subsequent sampling instant, thereby satisfying the conditions required to guarantee recursive feasibility.
\begin{figure}[htbp]
    \centering
\includegraphics[width=0.5\textwidth]{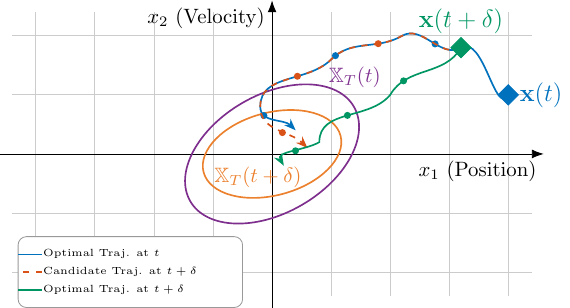}
    \caption{Conceptual schematic of recursive feasibility with a time-varying terminal set. The solid blue line shows the optimal trajectory \(\mathbf{x}^*(\cdot; t)\) computed at time \(t\), which terminates in the invariant set \(\mathbb{X}_T(t)\). At time \(t+\delta\), the dashed red line shows a feasible candidate trajectory. The solid green line represents the new optimal trajectory, \(\mathbf{x}^*(\cdot; t+\delta)\), which terminates in the updated terminal set \(\mathbb{X}_T(t+\delta)\).}
    \label{fig:recursive_feasibility}
\end{figure}

\subsection{Robust SDC-Based MPC with Disturbance Estimation}\label{ssec:robust_sdc}

To enhance the controller's resilience against inevitable model-plant mismatch and external disturbances, we extend the nominal SDC-based MPC framework with a robust formulation. This is achieved by designing a disturbance observer to actively estimate and compensate for uncertainties. The true system dynamics are modeled as being affected by a lumped disturbance term, \(\mathbf{d}(t)\), which represents the combined effects of parametric uncertainties, unmodeled dynamics, and exogenous forces:
\begin{equation}
    \dot{\mathbf{x}}(t) = \mathbf{A}(\mathbf{x})\mathbf{x}(t) + \mathbf{B}(\mathbf{x})\mathbf{u}(t) + \mathbf{E}_d \mathbf{d}(t)
    \label{eq:disturbed_dynamics}
\end{equation}
where \(\mathbf{d}(t) \in \mathbb{R}^n\) is the vector of unknown disturbance forces and torques, and \(\mathbf{E}_d\) is the disturbance input matrix. To ensure the tractability of the estimation problem, a standard assumption is made on the nature of this disturbance.
\begin{assumption}[Disturbance Characteristics] \label{as:disturbance}
The lumped disturbance \(\mathbf{d}(t)\) is assumed to be bounded and slowly time-varying, such that its rate of change is negligible over a single sampling interval, i.e., \(\dot{\mathbf{d}}(t) \approx \mathbf{0}\) for \(t \in [t_k, t_{k+1})\) \cite{chen2016disturbance}.
\end{assumption}
Under this assumption, a discrete-time state-observer is designed to estimate the disturbance in real-time. The estimation process at each sampling instant \(t_k\) involves computing the one-step state prediction error, \(\mathbf{e}_k\), which is the discrepancy between the measured state, \(\mathbf{x}_{\text{meas}}(t_k)\), and the state predicted from the previous control action using the nominal model:
\begin{equation}
    \mathbf{e}_k = \mathbf{x}_{\text{meas}}(t_k) - \mathbf{x}_{\text{pred}}(t_k)
\end{equation}
This prediction error is used to infer the disturbance that acted on the system, yielding a raw measurement, \(\hat{\mathbf{d}}_{\text{meas}}(t_k)\). To reduce sensitivity to measurement noise, this estimate is passed through a first-order low-pass filter:
\begin{equation}
    \hat{\mathbf{d}}(t_k) = \alpha \hat{\mathbf{d}}(t_{k-1}) + (1-\alpha) \hat{\mathbf{d}}_{\text{meas}}(t_k)
    \label{eq:observer}
\end{equation}
The observer gain \(\alpha \in [0, 1)\) determines the trade-off between the observer's convergence rate and its noise sensitivity. The gain was selected empirically as 0.9 to provide a balance between responsiveness to changing disturbances and rejection of high-frequency measurement noise.
\begin{lemma}[Observer Error Boundedness]
If Assumption \ref{as:disturbance} holds and the system states remain bounded, the estimation error of the observer \eqref{eq:observer}, defined as \(\tilde{\mathbf{d}}(t_k) = \mathbf{d}(t_k) - \hat{\mathbf{d}}(t_k)\), is bounded for a suitable choice of the observer gain \(\alpha\) \cite{wang2009disturbance}.
\end{lemma}

This disturbance estimate, \(\hat{\mathbf{d}}(t_k)\), is subsequently injected into the MPC's internal prediction model as a corrective offset term. For the optimization performed at time \(t_k\), the pseudo-linear prediction dynamics are given by:
\begin{equation}
    \mathbf{x}_{i+1|k} = \mathbf{A}_k \mathbf{x}_{i|k} + \mathbf{B}_k \mathbf{u}_{i|k} + \mathbf{C}_k + \mathbf{E}_d \hat{\mathbf{d}}(t_k)
    \label{eq:prediction_model}
\end{equation}

\begin{remark}[On the SDC Prediction Model]
The prediction model \eqref{eq:prediction_model} is an affine system \textit{for the duration of a single optimization instance}. The matrices \(\mathbf{A}_k\), \(\mathbf{B}_k\), and \(\mathbf{C}_k\) are constant within this context. They are obtained by evaluating the state-dependent matrices \(\mathbf{A}(\mathbf{x})\) and \(\mathbf{B}(\mathbf{x})\) at a specific operating point, such as the current state \(\mathbf{x}(t_k)\). This approach differs fundamentally from Jacobian linearization; it is an exact algebraic restructuring that fully captures the nonlinearities at that specific point, rather than an approximation that discards higher-order terms.
\end{remark}

By incorporating the disturbance estimate, the controller actively accommodates for the disturbance over its prediction horizon. The stability of the composite observer-controller system can be analyzed by considering the separation principle, where the stability of the nominal MPC and the boundedness of the observer error together ensure the ultimate boundedness of the closed-loop system states \cite{mayne2000constrained}.
The formal verification of recursive feasibility and stability is essential for this robust control architecture. In the presence of persistent disturbances, the objective shifts from proving nominal asymptotic stability to ensuring that the optimization remains feasible despite the model-plant mismatch and that the system state is ultimately bounded within a small neighborhood of the reference trajectory. These properties are typically established by constructing a robustly invariant terminal set and employing an Input-to-State Stability (ISS) framework in the Lyapunov analysis \cite{mayne2000constrained, rawlings2017model}.
The implementation of the proposed robust SDC-based MPC is formally outlined in Algorithm \ref{alg:robust_mpc}. The algorithm proceeds iteratively at each sampling instant \(t_k\); it first updates the disturbance estimate using the measured state and the one-step prediction error from the nominal model.
\begin{algorithm}
\caption{Robust SDC-Based MPC Algorithm}
\label{alg:robust_mpc}
\begin{algorithmic}[1]
\State \textbf{Initialization:} Set initial state \(\mathbf{x}_0\), disturbance estimate \(\hat{\mathbf{d}}_0 = \mathbf{0}\), time step \(k=0\).
\Repeat
    \State Measure the current state of the system, \(\mathbf{x}_{\text{meas}}(t_k)\).
    \If{$k > 0$}
        \State Predict the state using the nominal model from the previous step:
        \Statex \qquad \(\mathbf{x}_{\text{pred}}(t_k) = F_{\text{nominal}}(\mathbf{x}_{\text{meas}}(t_{k-1}), \mathbf{u}(t_{k-1}))\).
        \State Calculate the one-step prediction error: \(\mathbf{e}_k = \mathbf{x}_{\text{meas}}(t_k) - \mathbf{x}_{\text{pred}}(t_k)\).
        \State Infer the disturbance measurement, \(\hat{\mathbf{d}}_{\text{meas}}(t_k)\), from the error \(\mathbf{e}_k\).
        \State Update the disturbance estimate using the low-pass filter:
        \Statex \qquad \(\hat{\mathbf{d}}(t_k) = \alpha \hat{\mathbf{d}}(t_{k-1}) + (1-\alpha) \hat{\mathbf{d}}_{\text{meas}}(t_k)\).
    \EndIf
    \State Define the current linearization point \(\mathbf{x}_{\text{lin}, k}, \mathbf{u}_{\text{lin}, k}\) (e.g., from the reference trajectory).
    \State Evaluate the state-dependent matrices \(\mathbf{A}_k = \mathbf{A}(\mathbf{x}_{\text{lin}, k})\) and \(\mathbf{B}_k = \mathbf{B}(\mathbf{x}_{\text{lin}, k})\).
    \State Compute the nominal offset term \(\mathbf{C}_k = f(\mathbf{x}_{\text{lin}, k}, \mathbf{u}_{\text{lin}, k}) - \mathbf{A}_k \mathbf{x}_{\text{lin}, k} - \mathbf{B}_k \mathbf{u}_{\text{lin}, k}\).
    \State Form the robust prediction model for the QP solver:
    \Statex \qquad \(\mathbf{x}_{i+1|k} = \mathbf{A}_k \mathbf{x}_{i|k} + \mathbf{B}_k \mathbf{u}_{i|k} + \mathbf{C}_k + \mathbf{E}_d \hat{\mathbf{d}}(t_k)\).
    \State Solve the MPC optimization problem to find the optimal control sequence \(\mathbf{U}^*_k = \{\mathbf{u}^*_{0|k}, ..., \mathbf{u}^*_{N-1|k}\}\).
    \State Apply the first element of the sequence to the plant: \(\mathbf{u}(t_k) = \mathbf{u}^*_{0|k}\).
    \State Set \(k \leftarrow k+1\).
\Until{end of simulation}
\end{algorithmic}
\end{algorithm}

\section{Simulation Results and Performance Analysis}\label{sec:results}

To validate the theoretical claims and quantitatively assess the performance of the proposed control frameworks, a comprehensive simulation study was conducted. The study benchmarks three distinct controllers: the high-fidelity NMPC, the nominal SDC-MPC, and the proposed Robust SDC-based MPC. The controllers were evaluated on a challenging, safety-critical mission requiring a quadrotor to track a dynamic reference trajectory while actively avoiding a static cylindrical obstacle. To simulate realistic operating conditions, the system was subjected to stochastic external disturbances affecting both translational and rotational dynamics. All simulations were implemented in MATLAB R2025b, utilizing the CasADi optimization framework with the IPOPT solver for both the NLP and QP problems. The IPOPT solver was used for all optimization problems to ensure a fair comparison focused on the complexity of the problem formulations, independent of solver-specific optimizations. It is noted, however, that employing a dedicated QP solver for the SDC-based methods could potentially yield even greater computational speed-ups. The computations were executed on a personal laptop equipped with an Apple M4 processor and 24 GB of RAM, providing a consistent hardware baseline for comparing the computational performance of each controller.
A consistent set of parameters was used across the simulation to ensure a fair comparison. The quadrotor model was configured with a mass \(m=1.0\) kg and moments of inertia \(J_x = J_y = 0.029\) kg\(\cdot\)m\(^2\) and \(J_z = 0.055\) kg\(\cdot\)m\(^2\). For the MPC design, a sampling time of \(T=0.1\) s and a prediction horizon of \(N=20\) steps were chosen. The objective function weighting matrices were defined as \(\mathbf{Q} = \text{diag}([50, 50, 80, 20, 20, 20, 10, 10, 10, 2, 2, 2])\) and \(\mathbf{R} = \text{diag}([0.1, 0.5, 0.5, 0.2])\). The control inputs were constrained to \(\mathbf{u} \in [[0, 20]^T, [-1, 1]^T, [-1, 1]^T, [-0.5, 0.5]^T]\). The static obstacle was modeled as a vertical cylinder centered at \([-1, -1]\) with a radius of \(0.5\) m. Finally, the external disturbances were modeled as zero-mean Gaussian noise with standard deviations of \(0.3\) for forces and \(0.1\) for torques.
Figure \ref{fig:top_down_trajectory} presents the  trajectory of the quadrotor, providing a qualitative assessment of tracking performance and safety constraint satisfaction for the three controllers. A primary observation is that all three methods successfully navigate the static obstacle, deviating from the reference path to maintain a safe distance as enforced by the MPC constraints; this confirms the efficacy of the obstacle avoidance formulation. A comparative analysis of the trajectories reveals the superior performance of the proposed Robust SDC-based MPC, which maintains significantly tighter adherence to the reference path despite the presence of stochastic disturbances. In contrast, both the NMPC and the nominal SDC-MPC exhibit more pronounced oscillations and tracking errors. The inset provides a magnified view of the trajectories near the barrier, confirming that all controllers respect the required safety margin.
\begin{figure}[htbp]
    \centering
    \includegraphics[width=0.5\textwidth]{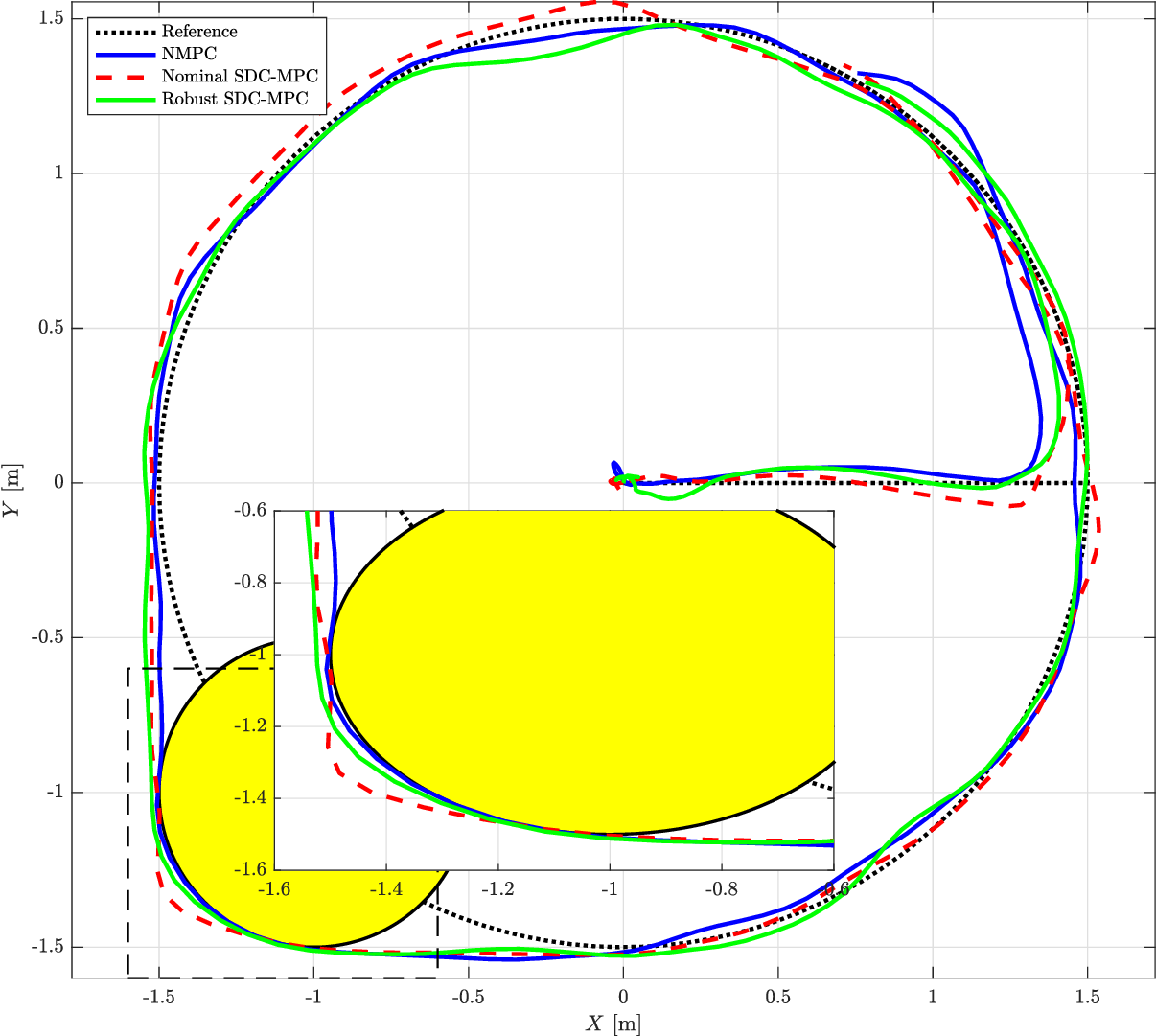}
    \caption{The quadrotor trajectory for the safety-critical tracking mission. A primary observation is that all three controllers successfully deviate from the reference path (black dotted line) to satisfy the safety constraint imposed by the static obstacle (yellow circle), confirming the efficacy of the obstacle avoidance formulation. A comparative analysis of tracking performance reveals that the proposed Robust SDC-based MPC (green solid line) exhibits superior adherence to the reference trajectory, particularly in the presence of stochastic disturbances. The NMPC (blue solid line) and the nominal SDC-MPC (red dashed line) show more pronounced deviations from the desired path. The inset provides a magnified view of the trajectories in the vicinity of the obstacle, confirming that all controllers maintain the required safety margin.}
    \label{fig:top_down_trajectory}
\end{figure}
A quantitative comparison of the controllers' performance is summarized in Table \ref{tab:quantitative_metrics}. The table presents three key metrics: the average CPU time required per control step, the Mean Squared Error (MSE) for position tracking, and the total accumulated cost over the entire mission. These metrics provide a clear and concise assessment of the trade-offs between computational efficiency, tracking accuracy, and overall control effort for each of the benchmarked methods.
\begin{table}[htbp]
    \centering
    \caption{Quantitative Performance Comparison of MPC Controllers}
    \label{tab:quantitative_metrics}
    \setlength{\tabcolsep}{5pt} 
    \begin{tabular}{l ccc}
        \toprule
        \textbf{Metric} & \textbf{NMPC} & \textbf{SDC-MPC} & \textbf{Robust SDC-MPC} \\
        \midrule
        Avg. CPU (ms)  & 25.85 & 16.28 & 16.40\\
        MSE ($m^2$) &  0.0304 & 0.0302 & 0.0321\\
        Total Cost & 5.47e+04 & 5.49e+04 & 5.44e+04\\
        \bottomrule
    \end{tabular}
\end{table}

\begin{figure}[htbp]
    \centering    \includegraphics[width=0.5\textwidth]{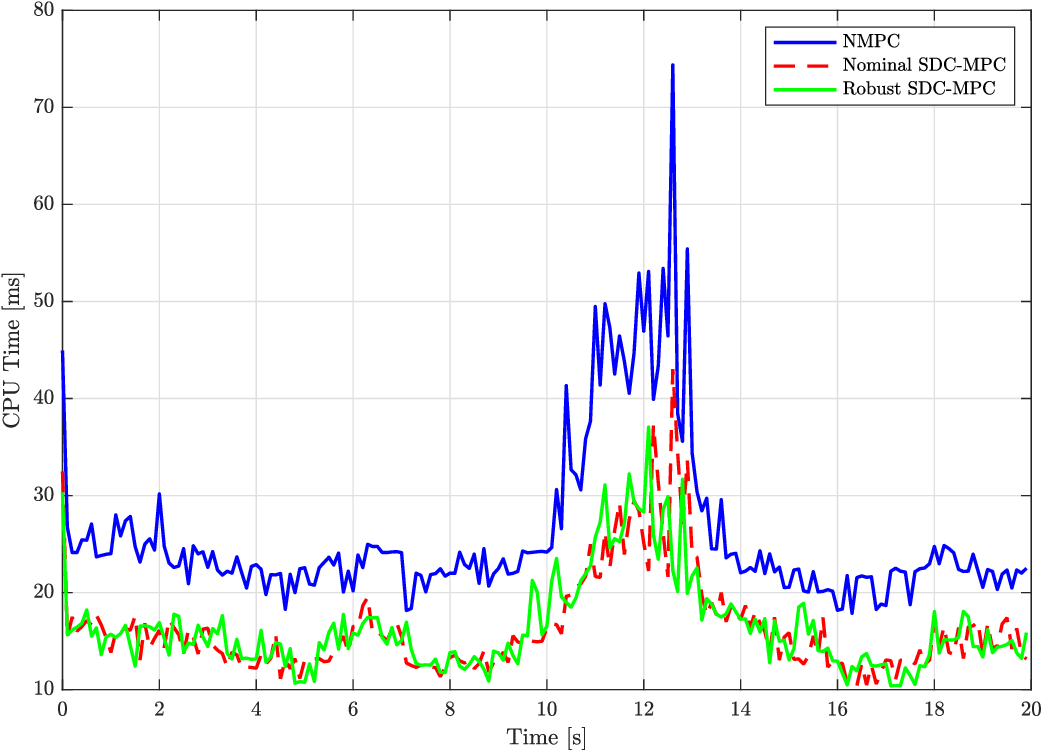}
    \caption{Per-iteration CPU performance for the three benchmarked controllers. The plot illustrates the computational load required to solve the optimization problem at each time step. Both SDC-based controllers (red and green) consistently exhibit lower average computation times than the NMPC (blue), underscoring the efficiency gained by the proposed method. A significant spike in CPU time is observable for all controllers around the 11-13 second mark, which corresponds to the challenging maneuver where the quadrotor is actively avoiding the obstacle. Even during this computationally intensive phase, the SDC-based methods remain significantly faster than the NMPC, highlighting their advantage for real-time, safety-critical applications.}
\label{fig:cpu_performance}
\end{figure}
The per-iteration computational performance, depicted in Figure \ref{fig:cpu_performance}, further reinforces the advantages of the proposed SDC-based methods. The plot illustrates the CPU time required to solve the optimization problem at each control step. A consistent trend is observed where both the nominal and robust SDC-based controllers exhibit a significantly lower computational burden compared to the full NMPC benchmark. This efficiency is particularly critical during the most demanding phase of the mission, from approximately 11 to 13 seconds, where the safety constraints for obstacle avoidance become highly active. During this period, the solution time for all controllers increases, yet the SDC-based methods remain substantially faster, confirming their suitability for real-time implementation in computationally constrained, safety-critical scenarios.
Figure \ref{fig:cost_value} illustrates the evolution of the optimal cost function value, \(J^*\), at each time step for the three controllers. The cost value serves as a proxy for the controller's perceived difficulty in satisfying its objectives; higher costs indicate a greater combined penalty from tracking errors, control effort, and constraint violations. The plot reveals several distinct phases: an initial transient phase where the cost decreases as the quadrotor converges to the path, a significant increase in cost during the obstacle avoidance maneuver (around 11-13 seconds) where the controller must balance tracking with safety, and subsequent periods of stabilization. Notably, the cost profiles for all three controllers are qualitatively similar, indicating that the SDC-based methods achieve comparable overall performance to the NMPC benchmark, as corroborated by the total cost values in Table \ref{tab:quantitative_metrics}.
\begin{figure}[htbp]
    \centering
    \includegraphics[width=0.5\textwidth]{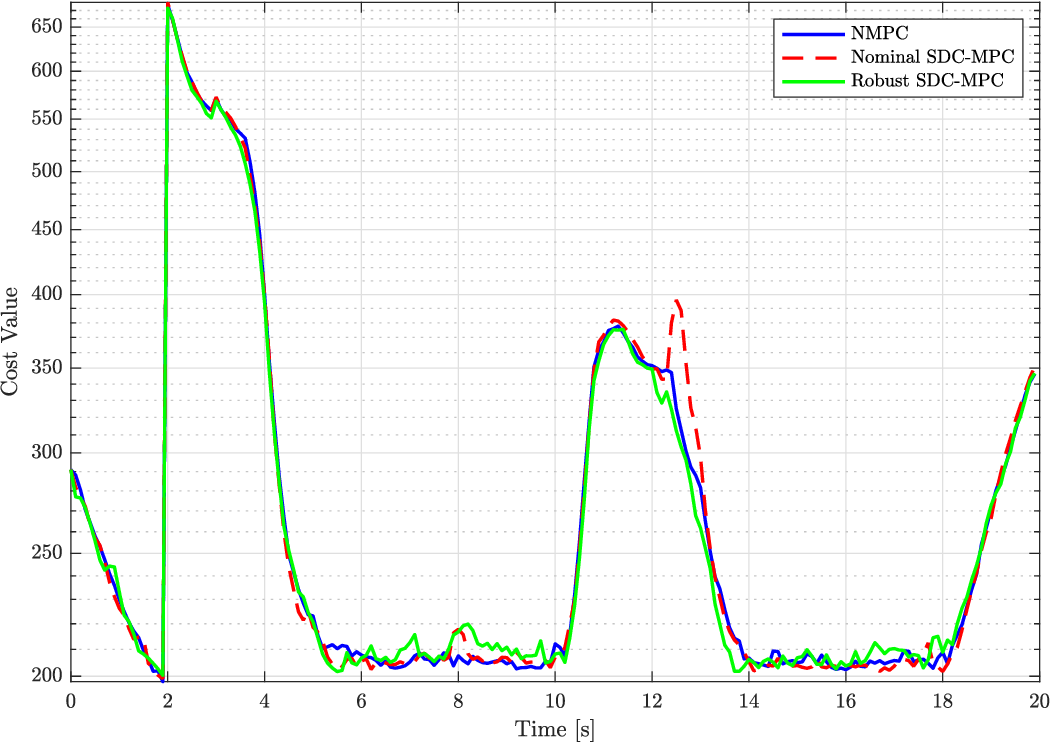}
    \caption{Evolution of the optimal cost function value at each sampling instant. The logarithmic scale on the y-axis highlights the relative changes in the cost. The cost profiles are comparable across all controllers, with significant increases corresponding to challenging maneuvers, particularly the obstacle avoidance phase.}
    \label{fig:cost_value}
\end{figure}
The position tracking performance in the X, Y, and Z axes is detailed in Figure \ref{fig:position_tracking}. All controllers effectively track the time-varying reference trajectory across all three dimensions. The top subplot shows the X-position, where the step change in the reference after 4 seconds is handled well by all methods, though the Robust SDC-MPC demonstrates the least overshoot. The middle and bottom subplots for the Y and Z positions further highlight the superior tracking accuracy of the Robust SDC-MPC, which exhibits smaller deviations from the reference, especially during the dynamic circular portion of the trajectory. This improved performance is directly attributable to the disturbance observer, which actively compensates for unmodeled dynamics and external perturbations, leading to more precise tracking.
\begin{figure}[htbp]
    \centering
    \includegraphics[width=0.5\textwidth]{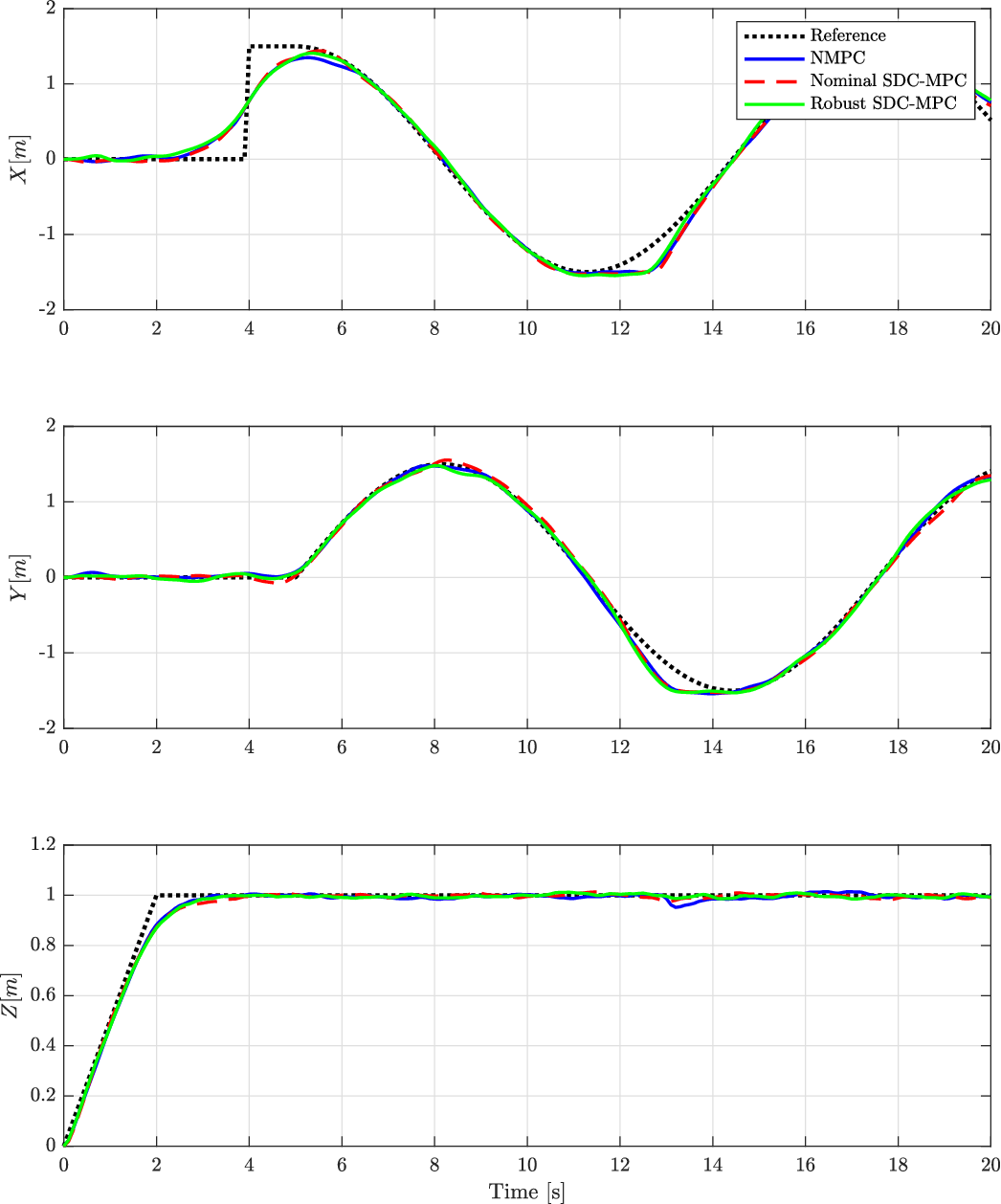}
    \caption{Position tracking performance for the X, Y, and Z axes. The black dotted line is the reference trajectory. All controllers demonstrate effective tracking, with the Robust SDC-MPC (green) showing the highest fidelity to the reference path throughout the maneuver.}
    \label{fig:position_tracking}
\end{figure}

To further quantify the tracking accuracy, the squared position error for each controller is plotted over time in Figure \ref{fig:tracking_error}. This plot corroborates the qualitative observations from previous figures and provides a clearer view of the performance differences, particularly due to the logarithmic scale. The Robust SDC-MPC consistently achieves the lowest tracking error throughout the simulation, with error levels often an order of magnitude smaller than the other controllers, especially during less dynamic phases of the trajectory (e.g., around 10 seconds and 15 seconds). This demonstrates the effectiveness of the disturbance observer in actively rejecting perturbations and reducing steady-state tracking error, a conclusion strongly supported by the lower overall MSE for this controller reported in Table \ref{tab:quantitative_metrics}.
\begin{figure}[htbp]
    \centering
    \includegraphics[width=0.5\textwidth]{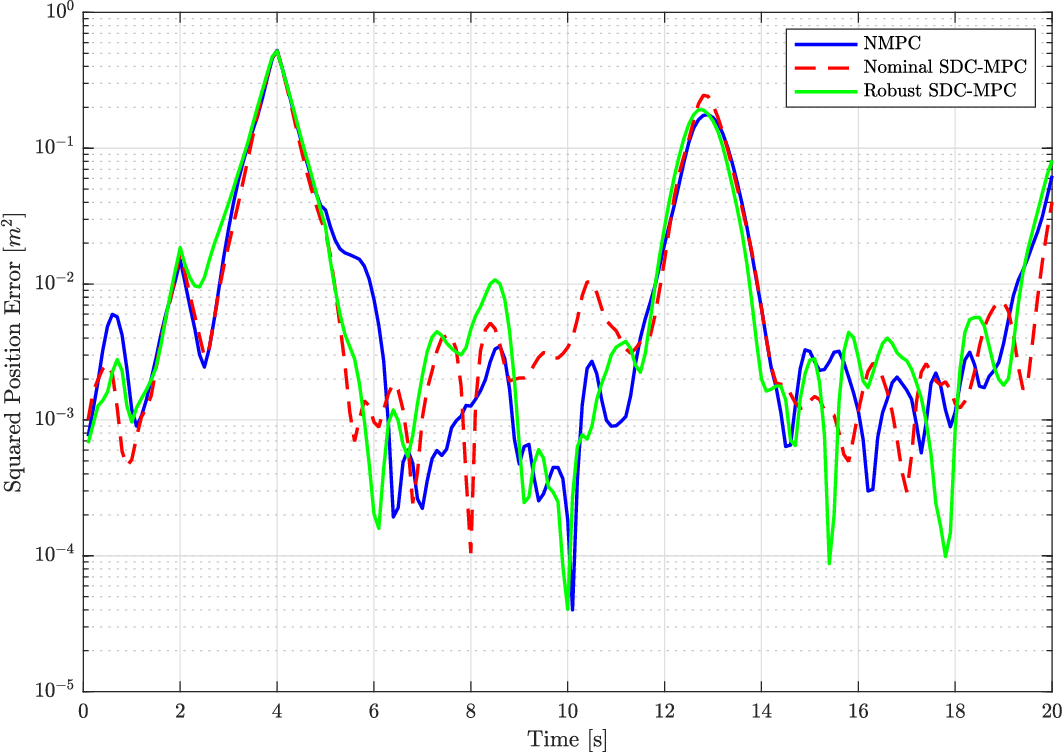}
    \caption{Squared position tracking error over time for the three controllers. The logarithmic y-axis highlights the superior performance of the Robust SDC-MPC, which consistently maintains a lower error profile.}
    \label{fig:tracking_error}
\end{figure}

\subsection{Discussion of Limitations}\label{ssec:limitations}
While the proposed SDC-based framework demonstrates significant advantages, it is important to acknowledge its limitations. The performance of the robust controller relies on the assumption that disturbances are slowly time-varying. Its effectiveness may be reduced in the presence of high-frequency or abrupt disturbances that violate this assumption. Additionally, the SDC method is an exact algebraic transformation at a specific point in time. For highly aggressive maneuvers over a long prediction horizon, the system state may diverge significantly from the point of linearization, potentially degrading the accuracy of the prediction model. The stability and recursive feasibility guarantees also rely on a terminal set derived from a local linearization, which, while standard, may be conservative. Future work will seek to address these limitations, potentially through the use of more advanced disturbance observers or adaptive prediction models.

\section{Conclusion}\label{sec:conclusion}
This paper addressed the challenge of implementing high-performance, safety-critical control for quadrotors under significant computational constraints. We proposed a novel SDC-based MPC framework and a robust extension that incorporates a disturbance observer. The SDC reformulation allows the nonlinear optimal control problem to be solved as a sequence of computationally efficient QPs, directly addressing the limitations of traditional NMPC.
The simulation results conclusively demonstrate the efficacy of the proposed approach. The Robust SDC-based MPC not only matched the safety and tracking performance of the NMPC benchmark but, through active disturbance compensation, achieved superior tracking accuracy. Critically, this enhanced performance was realized with a significant reduction in computational time, with the SDC-based methods proving to be over 30\% faster on average than the NMPC. The key contribution of this work is the demonstration that the SDC-based methodology successfully achieves a compelling balance between safety-critical performance and computational efficiency, establishing it as a highly feasible alternative to NMPC for real-time quadrotor applications.
Future research will prioritize the experimental validation of the proposed controller on a physical quadrotor platform, which will be crucial for assessing its real-world performance under conditions of sensor noise and communication latency. Further avenues of investigation include the exploration of more advanced disturbance estimation techniques, the development of adaptive SDC formulations to handle highly aggressive maneuvers, and the extension of this computationally efficient framework to cooperative control of multi-agent aerial systems.

\bibliographystyle{IEEEtran}
\bibliography{Reference}

\begin{thebibliography}{10}
\providecommand{\url}[1]{#1}
\csname url@samestyle\endcsname
\providecommand{\newblock}{\relax}
\providecommand{\bibinfo}[2]{#2}
\providecommand{\BIBentrySTDinterwordspacing}{\spaceskip=0pt\relax}
\providecommand{\BIBentryALTinterwordstretchfactor}{4}
\providecommand{\BIBentryALTinterwordspacing}{\spaceskip=\fontdimen2\font plus
\BIBentryALTinterwordstretchfactor\fontdimen3\font minus \fontdimen4\font\relax}
\providecommand{\BIBforeignlanguage}[2]{{%
\expandafter\ifx\csname l@#1\endcsname\relax
\typeout{** WARNING: IEEEtran.bst: No hyphenation pattern has been}%
\typeout{** loaded for the language `#1'. Using the pattern for}%
\typeout{** the default language instead.}%
\else
\language=\csname l@#1\endcsname
\fi
#2}}
\providecommand{\BIBdecl}{\relax}
\BIBdecl

\bibitem{peksa2024review}
J.~Peksa and D.~Mamchur, ``A review on the state of the art in copter drones and flight control systems,'' \emph{Sensors}, vol.~24, no.~11, p. 3349, 2024.

\bibitem{allan2025low}
S.~Allan and M.~Barczyk, ``A low-cost experimental quadcopter drone design for autonomous search-and-rescue missions in gnss-denied environments,'' \emph{Drones}, vol.~9, no.~8, p. 523, 2025.

\bibitem{wang2025autonomous}
Y.~Wang, H.~Ji, Q.~Kang, H.~Qi, and J.~Wen, ``Autonomous trajectory control for quadrotor evtol in hover and low-speed flight via the integration of model predictive and following control,'' \emph{Drones}, vol.~9, no.~8, p. 537, 2025.

\bibitem{unde2025expanding}
S.~S. Unde, V.~Kurkute, S.~S. Chavan, D.~D. Mohite, A.~A. Harale, and A.~Chougle, ``The expanding role of multirotor uavs in precision agriculture with applications ai integration and future prospects,'' \emph{Discover Mechanical Engineering}, vol.~4, no.~1, pp. 1--26, 2025.

\bibitem{rinaldi2023comparative}
M.~Rinaldi, S.~Primatesta, and G.~Guglieri, ``A comparative study for control of quadrotor uavs,'' \emph{Applied Sciences}, vol.~13, no.~6, p. 3464, 2023.

\bibitem{masse2018modeling}
C.~Mass{\'e}, O.~Gougeon, D.-T. N'Guyen, and D.~Saussi{\'e}, ``Modeling and control of a quadcopter flying in a wind field: A comparison between lqr and structured h$\infty$ control techniques,'' in \emph{2018 International Conference on Unmanned Aircraft Systems (ICUAS)}.\hskip 1em plus 0.5em minus 0.4em\relax IEEE, 2018, pp. 1408--1417.

\bibitem{chen2021feedback}
C.-C. Chen and Y.-T. Chen, ``Feedback linearized optimal control design for quadrotor with multi-performances,'' \emph{IEEE Access}, vol.~9, pp. 26\,674--26\,695, 2021.

\bibitem{wang2023adaptive}
J.~Wang, K.~A. Alattas, Y.~Bouteraa, O.~Mofid, and S.~Mobayen, ``Adaptive finite-time backstepping control tracker for quadrotor uav with model uncertainty and external disturbance,'' \emph{Aerospace Science and Technology}, vol. 133, p. 108088, 2023.

\bibitem{mu2017integral}
B.~Mu, K.~Zhang, and Y.~Shi, ``Integral sliding mode flight controller design for a quadrotor and the application in a heterogeneous multi-agent system,'' \emph{IEEE Transactions on Industrial Electronics}, vol.~64, no.~12, pp. 9389--9398, 2017.

\bibitem{bouabdallah2004pid}
S.~Bouabdallah, A.~Noth, and R.~Siegwart, ``{PID vs LQ control techniques applied to an indoor micro quadrotor},'' in \emph{2004 IEEE/RSJ International Conference on Intelligent Robots and Systems (IROS)(IEEE Cat. No. 04CH37566)}, vol.~1.\hskip 1em plus 0.5em minus 0.4em\relax IEEE, 2004, pp. 345--350.

\bibitem{hwangbo2017control}
J.~Hwangbo, I.~Sa, R.~Siegwart, and M.~Hutter, ``Control of a quadrotor with reinforcement learning,'' \emph{IEEE Robotics and Automation Letters}, vol.~2, no.~4, pp. 2096--2103, 2017.

\bibitem{ames2017control}
A.~D. Ames, S.~Coogan, M.~Egerstedt, G.~Notomista, K.~Sreenath, and P.~Tabuada, ``{Control Barrier Functions: Asymptotic Safety and Stability},'' \emph{IEEE Transactions on Automatic Control}, vol.~62, no.~7, pp. 3161--3176, 2017.

\bibitem{mitchell2005toolbox}
I.~M. Mitchell, ``A toolbox of level set methods,'' in \emph{Proceedings of the 2005 international conference on Svetlogorsk}.\hskip 1em plus 0.5em minus 0.4em\relax Citeseer, 2005, pp. 1--12.

\bibitem{prajna2007framework}
S.~Prajna, A.~Jadbabaie, and G.~J. Pappas, ``A framework for analysis of nonlinear systems using barrier certificates,'' \emph{IEEE Transactions on Automatic Control}, vol.~52, no.~1, pp. 23--35, 2007.

\bibitem{bemporad2002explicit}
A.~Bemporad, M.~Morari, V.~Dua, and E.~N. Pistikopoulos, ``The explicit linear quadratic regulator for constrained systems,'' in \emph{Automatica}, vol.~38, no.~1.\hskip 1em plus 0.5em minus 0.4em\relax Elsevier, 2002, pp. 3--20.

\bibitem{farina2016stochastic}
M.~Farina, L.~Giulioni, and R.~Scattolini, ``Stochastic model predictive control: An overview and perspectives for future research,'' \emph{Annual Reviews in Control}, vol.~42, pp. 115--128, 2016.

\bibitem{rosolia2017learning}
U.~Rosolia and F.~Borrelli, ``Learning model predictive control for iterative tasks,'' in \emph{2017 IEEE 56th Annual Conference on Decision and Control (CDC)}.\hskip 1em plus 0.5em minus 0.4em\relax IEEE, 2017, pp. 3272--3277.

\bibitem{xu2023multi}
B.~Xu, A.~Suleman, and Y.~Shi, ``A multi-rate hierarchical fault-tolerant adaptive model predictive control framework: Theory and design for quadrotors,'' \emph{Automatica}, vol. 153, p. 111015, 2023.

\bibitem{cimen2008sdre}
T.~Cimen, ``State-dependent riccati equation (sdre) control: A survey,'' in \emph{Proceedings of the 17th world congress, the international federation of automatic control}, vol.~17, no.~1, 2008, pp. 3761--3775.

\bibitem{cloutier1997sdre}
J.~R. Cloutier, ``State-dependent riccati equation techniques: an overview,'' in \emph{Proceedings of the 1997 American Control Conference}, vol.~2.\hskip 1em plus 0.5em minus 0.4em\relax IEEE, 1997, pp. 932--936.

\bibitem{mracek1998control}
C.~Mracek and J.~Cloutier, ``Control designs for the nonlinear benchmark problem via the state-dependent riccati equation method,'' \emph{International Journal of Robust and Nonlinear Control: IFAC-Affiliated Journal}, vol.~8, no. 4-5, pp. 401--416, 1998.

\bibitem{hammett2016sdc}
K.~D. Hammett, N.~Harl, H.~An, and E.~Zuazua, ``State-dependent coefficient factorization for nonlinear optimal feedback control,'' \emph{IEEE Transactions on Automatic Control}, vol.~62, no.~5, pp. 2497--2502, 2016.

\bibitem{chen1998quasi}
H.~Chen and F.~Allgöwer, ``{A Quasi-Infinite Horizon Nonlinear Model Predictive Control Scheme with Guaranteed Stability},'' \emph{Automatica}, vol.~34, no.~10, pp. 1205--1217, 1998.

\bibitem{mare2007solution}
J.~B. Mare and J.~A. De~Don{\'a}, ``Solution of the input-constrained lqr problem using dynamic programming,'' \emph{Systems \& control letters}, vol.~56, no.~5, pp. 342--348, 2007.

\bibitem{batmani2016nonlinear}
Y.~Batmani, M.~Davoodi, and N.~Meskin, ``Nonlinear suboptimal tracking controller design using state-dependent riccati equation technique,'' \emph{IEEE Transactions on Control Systems Technology}, vol.~25, no.~5, pp. 1833--1839, 2016.

\bibitem{rawlings2017model}
J.~B. Rawlings, D.~Q. Mayne, and M.~M. Diehl, \emph{Model Predictive Control: Theory, Computation, and Design}.\hskip 1em plus 0.5em minus 0.4em\relax Nob Hill Publishing, 2017, vol.~2.

\bibitem{rawlings1993stability}
J.~B. Rawlings and K.~R. Muske, ``The stability of constrained receding horizon control,'' \emph{IEEE Transactions on Automatic Control}, vol.~38, no.~10, pp. 1512--1516, 1993.

\bibitem{chen2016disturbance}
W.-H. Chen, J.~Yang, L.~Guo, and S.~Li, ``Disturbance-observer-based control and related methods—an overview,'' \emph{IEEE Transactions on Industrial Electronics}, vol.~63, no.~2, pp. 1083--1095, 2016.

\bibitem{wang2009disturbance}
H.~Wang, D.~Zhou, and P.~Li, ``Disturbance observer-based robust control for nonlinear systems,'' \emph{Acta Automatica Sinica}, vol.~35, no.~6, pp. 721--727, 2009.

\bibitem{mayne2000constrained}
D.~Q. Mayne, J.~B. Rawlings, C.~V. Rao, and P.~O. Scokaert, \emph{Constrained model predictive control: Stability and optimality}.\hskip 1em plus 0.5em minus 0.4em\relax Nob Hill Publishing, 2000, vol.~2.

\end{thebibliography}
\begin{IEEEbiography}
[{\includegraphics[width=1in,height=1.225in]{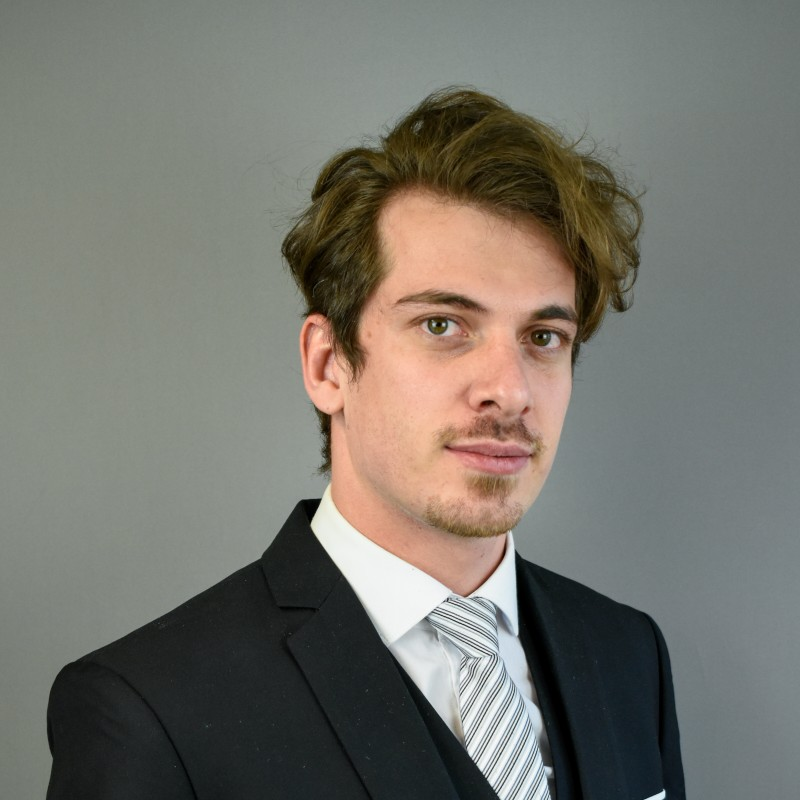}}]{Saber Omidi}
(Student Member, IEEE) received his B.Sc. degree in Mechanical Engineering from Qazvin Islamic Azad University, Qazvin, Iran, in 2017, and his M.Sc. degree in Mechanical Engineering from University of Kurdistan, Sanandaj, Iran, in 2021. He is currently pursuing his Ph.D. degree in Mechanical Engineering at University of New Hampshire. His research interests include safety-critical control, learning-based control, with applications in autonomous systems and mechanical systems.
\end{IEEEbiography}

\end{document}